\documentclass[11pt,a4paper]{article}
\usepackage[margin=1in]{geometry}
\usepackage{libertine}
\usepackage[T1]{fontenc}
\usepackage[utf8]{inputenc}
\usepackage{amsfonts,amssymb,amsmath,amsthm}
\usepackage{calc}
\usepackage{tikz}
\usetikzlibrary{calc}
\usepackage[bottom,para,symbol]{footmisc}
\usepackage{hyperref}
\usepackage{thmtools}
\usepackage{thm-restate}

\usepackage{pdfpages}
\usepackage{epigraph}

\usepackage{url}
\usepackage{hyperref}
\usepackage{enumitem}
\usepackage{xspace}
\usepackage{todonotes}
\usepackage{thmtools}
\usepackage{mathrsfs}  
\usepackage{textpos}
\usepackage[normalem]{ulem}
\usepackage[capitalize, nameinlink]{cleveref}

\DeclareMathOperator{\indeg}{indeg}
\DeclareMathOperator{\outdeg}{outdeg}

\newtheorem{theorem}{Theorem}[section]

\newtheorem{claim}[theorem]{Claim}

\theoremstyle{definition}

\renewcommand{\geq}{\geqslant}

\usepackage{framed}
\usepackage{tabularx}

\newlength{\RoundedBoxWidth}
\newsavebox{\GrayRoundedBox}
\newenvironment{GrayBox}[1]%
   {\setlength{\RoundedBoxWidth}{.93\columnwidth}
    \def\boxheading{#1}
    \begin{lrbox}{\GrayRoundedBox}
       \begin{minipage}{\RoundedBoxWidth}}%
   {   \end{minipage}
    \end{lrbox}
    \begin{center}
    \begin{tikzpicture}%
       \node(Text)[draw=black!20,fill=white,rounded corners,inner sep=2ex,text width=\RoundedBoxWidth]
             {\usebox{\GrayRoundedBox}};
        \coordinate(x) at (current bounding box.north west);
        \node [draw=white,rectangle,inner sep=3pt,anchor=north west,fill=white]
        at ($(x)+(6pt,.75em)$) {\boxheading};
    \end{tikzpicture}
    \end{center}}

\newenvironment{defproblemx}[1]{\noindent\ignorespaces%
                                \FrameSep=6pt%
                                \parindent=0pt%
                \begin{GrayBox}{#1}%
                \begin{tabular*}{\columnwidth}{!{\extracolsep{\fill}}@{\hspace{.1em}} >{\itshape} p{1.5cm} p{0.86\columnwidth} @{}}%
            }{
                \end{tabular*}%
                \end{GrayBox}%
                \ignorespacesafterend
            }

\usepackage{mathtools}

\title{A note on finding long directed cycles above the minimum degree bound
in 2-connected digraphs}

\author{
\and Jadwiga Czyżewska\thanks{University of Warsaw, Poland (\texttt{j.czyzewska@mimuw.edu.pl}).
Supported by Polish National Science Centre SONATA BIS-12 grant number 2022/46/E/ST6/00143.}
\and Marcin Pilipczuk\thanks{University of Warsaw, Poland (\texttt{m.pilipczuk@mimuw.edu.pl}).
Supported by Polish National Science Centre SONATA BIS-12 grant number 2022/46/E/ST6/00143.}}

\begin{document}
\date{}
\maketitle

\begin{abstract}
For a directed graph $G$, let $\mathrm{mindeg}(G)$ be the minimum 
among in-degrees and out-degrees of all vertices of $G$. 
It is easy to see that $G$ contains a directed cycle of length at least $\mathrm{mindeg}(G)+1$.
In this note, we show that, even if $G$ is $2$-connected, it is NP-hard to check
if $G$ contains a cycle of length at least $\mathrm{mindeg}(G)+3$. 
This is in contrast with recent algorithmic 
results of Fomin, Golovach, Sagunov, and Simonov [SODA 2022] for
analogous questions in undirected graphs. 
\end{abstract}

\section{Introduction}
The \emph{minimum degree} of a graph $G$, denoted $\mathrm{mindeg}(G)$, 
is defined as follows: if $G$ is undirected, it is just the minimum among the degrees
of all vertices of $G$ and if $G$ is directed, it is the minimum  of in-degree and out-degree
among all vertices of $G$. 

Dirac in 1952~\cite{dirac1952} proved that an undirected 2-connected graph
contains a cycle of length at least $2\mathrm{mindeg}(G)$. 
At SODA 2022, Fomin, Golovach, Sagunov, and Simonov~\cite{DBLP:conf/soda/FominGSS22} showed the following
algorithmic version of Dirac's theorem: finding a cycle of length
at least $2\mathrm{mindeg}(G)+k$ is fixed-parameter tractable in $k$,
that is, solvable in time $f(k) \cdot \mathrm{poly}(|V(G)|)$ for some computable
function $f$. 
A natural question, repeated informally in a number of open problem sessions in the last few
years (and also in the conclusions section of~\cite{DBLP:conf/soda/FominGSS22}, albeit with an incorrect
recollection of a theorem of Thomassen~\cite{thomassen1981}), is whether there is any analog of
this fact in directed graphs.

Thomassen~\cite{thomassen1981} showed that the analog of Dirac's theorem is true in the special case
$|V(G)| = 2\mathrm{mindeg}(G)+1$.
It is easy to see that every directed graph $G$ contains a cycle of length at least
$\mathrm{mindeg}(G)+1$. It is then natural to ask about parameterized algorithms
for finding cycles of length $\mathrm{mindeg}(G)+k$. 
In this work, we show that this problem is NP-hard for $k=3$ and $G$ being $2$-connected.
(A directed graph is $2$-connected if, for every two vertices $u$ and $v$, there exist
two vertex-disjoint paths from $u$ to $v$ in $G$.)

Observe that the 2-connectivity assumption is essential due to the following construction
(that also appears in the undirected setting in~\cite{DBLP:journals/siamdm/FominGSS24}). 
Let $G$ be an instance of the \textsc{Directed Hamiltonian Cycle} problem and
denote $n = |V(G)|$. For every $v \in V(G)$, create a directed clique $K_v$ on $n-1$
vertices and identify one vertex with $v$. The new graph has minimum degree $n-2$
and the only hope for a~cycle of length at least $n$ is to have a Hamiltonian cycle
in the input graph. 

Our work is essentially an observation that an analogous degree-inflating gadget exists for edges
of the input \textsc{Hamiltonian Cycle} instance
(instead of vertices as in the aforementioned construction), keeping the final
graph $2$-connected. 
It remains open whether higher connectivity assumption (such as, say, $3$-connectivity)
changes the picture significantly.

\paragraph{Notation.} We use standard graph notation. All graphs in this paper are finite, simple and unweighted. The set of vertices of graph $G$ is denoted as $V(G)$ and the set of edges as $E(G)$. 

An edge between vertices $u$ and $u$ in an undirected graph is denoted as $uv$. An arc from vertex $u$ to $v$ in a directed graph is denoted as $(u,v)$.

The degree of a vertex $v$ is denoted as $\deg(v)$. In directed graphs we denote the indegree and outdegree as $\indeg(v)$ and $\outdeg(v)$, respectively.

By $X_1-X_2-\dots-X_k$ we denote a path $v_1-v_2-\dots-v_k$, where $v_i$ is some vertex belonging to $X_i$ for each $i = 1,2,\dots, k$. If $X_i$ is singleton $\{x_i\}$ for some $i$, we write $-x_i-$ instead of $-\{x_i\}-$.

A cycle $C$ of length $m$ is a sequence of vertices $v_1-v_2- \dots- v_m- v_{m+1}$, where $v_{m+1}=v_1$ but otherwise the vertices $v_i$ are pairwise distinct, and $v_iv_{i+1}\in E(G)$ for every $i=1,2,\dots, m$. 

A directed clique $G$ is a directed graph such that for every pair of vertices $u,v\in V(G)$ arcs $(u,v)$ and $(v,u)$ belong to $E(G)$.

\section{Main proof}

\begin{theorem}
    For every integer $a \geq 3$, the following problem is NP-hard:
    given a directed $2$-connected graph $G$, accept if $G$ contains a cycle
    of length at least $\mathrm{mindeg}(G)+a$. 
\end{theorem}
\begin{proof}
    We reduce from \textsc{Undirected Hamiltonian Cycle}. Just to recall, in the \textsc{Undirected Hamiltonian Cycle} problem we are given on input an undirected graph $G$ and we are to decide whether there exists a cycle of length $V(G)$ in $G$. Let $n=|V(G)|$.
    Our strategy is to create a new directed graph $H$ out of $G$ such that each each vertex in $H$ has in-degree and out-degree at least $2n-a$ and we ask whether $H$ contains a cycle of length at least $2n$.
    
    Without loss of generality, we can assume that the input graph $G$ is $2$-connected,
    as otherwise the input instance is clearly a no-instance.
    Also, we assume $|V(G)| \geq a$.

    Given a graph $G$, we construct a new graph $H$ in the following manner. First, we take $V(H)=V(G)$. Then, for each $uv\in E(G)$ we add $2n$ directed cliques $C_{uv}^i$, $i =1,2,\ldots,2n$, each on $2n-a+1$ vertices. Moreover, for every $i=1,2,\ldots,2n$, we choose two vertices $u',v'$ of $C_{uv}^i$ and add arcs $(u,u'),\ (u',v),\ (v,v'),\ (v',u)$ to the graph $H$. Vertices belonging to $C_{uv}^i$ for some $uv$ and $i$ are called {\it clique vertices}. This finishes the description of the graph $H$.

    Observe that every clique vertex has in- and out-degree $2n-a$ or $2n-a+1$, while every
    vertex of $G$ has in- and out-degree at least $2n$ in $H$. 
    Furthermore, it is easy to check that $H$ is $2$-connected, as $G$ is $2$-connected.

    \begin{claim}
        If there exists a cycle $C$ in $G$ of length $n$, then there exists a cycle $C'$ in $H$ of length at least $2n$.
    \end{claim}
    \begin{proof}
        Let $C=v_1-v_2-\dots- v_n-v_{n+1}$, where $v_{n+1}=v_1$ and $v_iv_{i+1}\in E(G)$ for every $i=1,2,\dots, n$. Then $v_1-C_{v_1v_2}^1-v_2-C_{v_2v_3}^1-v_3-\dots-v_n-C_{v_nv_{n+1}}^1-v_{n+1}$ is a cycle in $H$ of length $2n$.
    \end{proof}

    \begin{claim}
        If there exists a cycle $C'$ in $H$ of length at least $2n$, then there exists a cycle $C$ in $G$ of length at least $n$.
    \end{claim}
    \begin{proof}
        Suppose that there exists $C_{uv}^i$ such that $C$ contains at least two vertices of $C_{uv}^i$. Then $C$ contains at most one vertex of $V(G)$ and therefore $C$ has length at most $(2n-a+1) + 1=2n-a+2 < 2n$ as $a \geq 3$. Thus for every $C_{uv}^i$ the cycle $C$ contains at most one vertex belonging to it and for each such vertex both its neighbors on the cycle belong to $V(G)$.
        By the construction of $H$ it means that these two vertices are connected in $G$. Therefore, by removing clique vertices and recovering edges of $G$ we decreased the length of the cycle by at most half. Thus there exists a cycle in length at least $\frac12 \cdot 2n=n$.
    \end{proof}
\end{proof}

\bibliographystyle{plain}
\bibliography{bibliography}

\end{document}